\documentclass[11pt]{article}%
\usepackage{amsmath}
\usepackage{amsfonts}
\usepackage{amssymb}
\usepackage{graphicx}
\usepackage{fullpage}%
\providecommand{\U}[1]{\protect\rule{.1in}{.1in}}
\newenvironment{proof}[1][Proof]{\noindent\textbf{#1.} }{\ \rule{0.5em}{0.5em}}
\newtheorem{theorem}{Theorem}

\newtheorem{lemma}[theorem]{Lemma}

\newtheorem{proposition}[theorem]{Proposition}

\begin{document}

\title{Advice Coins for Classical and Quantum Computation}
\author{Scott Aaronson\thanks{MIT. \ Email: aaronson@csail.mit.edu. \ \ This material
is based upon work supported by the National Science Foundation under Grant
No. 0844626. \ Also supported by a DARPA YFA grant and a Sloan Fellowship.}
\and Andrew Drucker\thanks{MIT. \ Email: adrucker@mit.edu. \ Supported by a DARPA
YFA grant.}}
\date{}
\maketitle

\begin{abstract}
We study the power of classical and quantum algorithms equipped with
nonuniform advice, in the form of a coin whose bias encodes useful
information. This question takes on particular importance in the quantum case,
due to a surprising result that we prove: \textit{a quantum finite automaton
with just two states can be sensitive to arbitrarily small changes in a coin's
bias}. \ This contrasts with classical probabilistic finite automata, whose
sensitivity to changes in a coin's bias is bounded by a classic 1970 result of
Hellman and Cover.

Despite this finding, we are able to bound the power of advice coins for
space-bounded classical and quantum computation. \ We define the classes
$\mathsf{BPPSPACE/coin}$ and $\mathsf{BQPSPACE/coin}$, of languages decidable
by classical and quantum polynomial-space machines with advice coins. \ Our
main theorem is that both classes coincide with $\mathsf{PSPACE/poly}$.
\ Proving this result turns out to require substantial machinery. \ We use an
algorithm due to Neff for finding roots of polynomials in $\mathsf{NC}$; a
result from algebraic geometry that lower-bounds the separation of a
polynomial's roots; and a result on fixed-points of superoperators due to
Aaronson and Watrous, originally proved in the context of quantum computing
with closed timelike curves.

\end{abstract}

\section{Introduction\label{INTRO}}

\subsection{The Distinguishing Problem}

The fundamental task of mathematical statistics is to learn features of a
random process from empirical data generated by that process. \ One of the
simplest, yet most important, examples concerns a coin with unknown bias.
\ Say we are given a coin which lands \textquotedblleft
heads\textquotedblright\ with some unknown probability $q$ (called the
\textit{bias}). \ In the \textit{distinguishing problem}, we assume $q$ is
equal either to $p$ or to $p+\varepsilon$, for some known $p,\varepsilon$, and
we want to decide which holds.

A traditional focus is the \textit{sample complexity} of statistical learning
procedures. \ For example, if $p=1/2$, then $t=\Theta\left(  \log\left(
1/\delta\right)  /\varepsilon^{2}\right)  $ coin flips are necessary and
sufficient to succeed with probability $1-\delta$ on the distinguishing
problem above. \ This assumes, however, that we are able to count the number
of heads seen, which requires $\log(t)$ bits of memory. From the perspective
of computational efficiency, it is natural to wonder whether methods with a
much smaller space requirement are possible. \ This question was studied in a
classic 1970 paper by Hellman and Cover~\cite{hc}. \ They showed that any
(classical, probabilistic) finite automaton that distinguishes with bounded
error between a coin of bias $p$ and a coin of bias $p+\varepsilon$, must have
$\Omega\left(  p\left(  1-p\right)  /\varepsilon\right)  $
states.\footnote{For a formal statement, see Section~\ref{HCSEC}.} \ Their
result holds with \textit{no restriction} on the number of coin flips
performed by the automaton. \ This makes the result especially interesting, as
it is not immediately clear how sensitive such machines can be to small
changes in the bias.

Several variations of the distinguishing problem for space-bounded automata
were studied in related works by Hellman~\cite{hellman_thesis} and
Cover~\cite{cover}. \ Very recently, Braverman, Rao, Raz, and
Yehudayoff~\cite{brry} and Brody and Verbin~\cite{bv:coin} studied the power
of restricted-width, \textit{read-once branching programs} for this problem.
\ The distinguishing problem is also closely related to the
\textit{approximate majority}\ problem, in which given an $n$-bit string $x$,
we want to decide whether $x$ has Hamming weight less than $\left(
1/2-\varepsilon\right)  n$ or more than $\left(  1/2+\varepsilon\right)  n$.
\ A large body of research has addressed the ability of constant-depth
circuits to solve the approximate majority problem and its
variants~\cite{aar:ph, ajtai83, amano, ow, viola:maj, viola10}.

\subsection{The Quantum Case}

In this paper, our first contribution is to investigate the power of
\textit{quantum} space-bounded algorithms to solve the distinguishing problem.
\ We prove the surprising result that, in the absence of noise, quantum finite
automata with a constant number of states can be sensitive to
\textit{arbitrarily small} changes in bias:

\begin{theorem}
[Informal]\label{qmnhc} For any $p\in\left[  0,1\right]  $ and $\varepsilon
>0$, there is a quantum finite automaton $M_{p,\varepsilon}$ with just two
states (not counting the $\left\vert \operatorname*{Accept}\right\rangle $ and
$\left\vert \operatorname*{Reject}\right\rangle $ states) that distinguishes a
coin of bias $p$ from a coin of bias $p+\varepsilon$; the difference in
acceptance probabilities between the two cases is at least $0.01$. \ (This
difference can be amplified using more states.)
\end{theorem}

In other words, the lower bound of Hellman and Cover \cite{hc}\ has no
analogue for quantum finite automata. \ The upshot is that we obtain a natural
example of a task that a quantum finite automaton can solve using
\textit{arbitrarily} fewer states than a probabilistic finite automaton, not
merely exponentially fewer states! \ Galvao and Hardy~\cite{galvao} gave a
related example, involving an automaton that moves continuously through a
field $\varphi$, and needs to decide whether an integral $\int_{0}^{1}%
\varphi\left(  x\right)  dx$ is odd or even, promised that it is an integer.
\ Here, a quantum automaton needs only a single qubit, whereas a classical
automaton cannot guarantee success with any finite number of bits.
\ Naturally, both our quantum automaton and that of Galvao and Hardy only work
in the absence of noise.

\subsection{Coins as Advice}

This unexpected power of quantum finite automata invites us to think further
about what sorts of statistical learning are possible using a small number of
qubits. \ In particular, if space-bounded quantum algorithms can detect
arbitrarily small changes in a coin's bias, then could a $p$-biased coin be an
incredibly-powerful \textit{information resource} for quantum computation, if
the bias $p$ was well-chosen? \ A bias $p\in(0,1)$ can be viewed in its binary
expansion $p=0.p_{1}p_{2}\ldots$\ as an infinite sequence of bits; by flipping
a $p$-biased coin, we could hope to access those bits, perhaps to help us
perform computations. \

This idea can be seen in \textquotedblleft Buffon's needle,\textquotedblright%
\ a probabilistic experiment that in principle allows one to calculate the
digits of $\pi$ to any desired accuracy.\footnote{See
http://en.wikipedia.org/wiki/Buffon\%27s\_needle} \ It can also be seen in the
old speculation that computationally-useful information might somehow be
encoded in dimensionless physical constants, such as the fine-structure
constant $\alpha\approx0.0072973525377$ that characterizes the strength of the
electromagnetic interaction. \ But leaving aside the question of which biases
$p\in\lbrack0,1]$ can be realized by actual physical processes, let us assume
that coins of \textit{any} desired bias are available. \ We can then ask: what
computational problems can be solved efficiently using such coins? \ This
question was raised to us by Erik Demaine (personal communication), and was
initially motivated by a problem in computational genetics. \

In the model that we use, a Turing machine receives an input $x$ and is given
access to a sequence of bits drawn independently from an \textit{advice coin}
with some arbitrary bias $p_{n}\in\lbrack0,1]$, which may depend on the input
length $n=\left\vert x\right\vert $. \ The machine is supposed to decide (with
high success probability) whether $x$ is in some language $L$. \ We allow
$p_{n}$ to depend only on $\left\vert x\right\vert $, not on $x$ itself, since
otherwise the bias could be set to $0$ or $1$ depending on whether $x\in L$,
allowing membership in $L$ to be decided trivially. \ We let
$\mathsf{BPPSPACE/coin}$ be the class of languages decidable with bounded
error by polynomial-space algorithms with an advice coin. \ Similarly,
$\mathsf{BQPSPACE/coin}$ is the corresponding class for polynomial-space
quantum algorithms. \ We impose no bound on the running time of these algorithms.

It is natural to compare these classes with the corresponding classes
$\mathsf{BPPSPACE/poly}$\ and $\mathsf{BQPSPACE/poly}$, which consist of all
languages decidable by $\mathsf{BPPSPACE}$\ and $\mathsf{BQPSPACE}$\ machines
respectively, with the help of an arbitrary \textit{advice string} $w_{n}%
\in\{0,1\}^{\ast}$ that can depend only on the input length $n=\left\vert
x\right\vert $. \ Compared to the standard advice classes, the strength of the
coin model is that an advice coin bias $p_{n}$ can be an arbitrary real
number, and so encode infinitely many bits; the weakness is that this
information is only accessible indirectly through the observed outcomes of
coin flips.

It is tempting to try to simulate an advice coin using a conventional advice
string, which simply specifies the coin's bias to $\operatorname*{poly}\left(
n\right)  $ bits of precision. \ At least in the classical case, the effect of
\textquotedblleft rounding\textquotedblright\ the bias can then be bounded by
the Hellman-Cover Theorem. \ Unfortunately, that theorem (whose bound is
essentially tight) is not strong enough to make this work: if the bias $p$ is
extremely close to $0$ or $1$, then a $\mathsf{PSPACE}$\ machine really
\textit{can} detect changes in $p$\ much smaller than
$2^{-\operatorname*{poly}\left(  n\right)  }$. \ This means that
upper-bounding the power of advice coins is a nontrivial problem even in the
classical case. \ In the quantum case, the situation is even worse, since as
mentioned earlier, the quantum analogue of the Hellman-Cover Theorem is false.

Despite these difficulties, we are able to show strong limits on the power of
advice coins in both the classical and quantum cases. \ Our main theorem says
that $\mathsf{PSPACE}$\ machines can effectively extract only
$\operatorname*{poly}\left(  n\right)  $ bits of \textquotedblleft useful
information\textquotedblright\ from an advice coin:

\begin{theorem}
[Main]\label{mainthm}$\mathsf{BQPSPACE/coin}=\mathsf{BPPSPACE/coin}%
=\mathsf{PSPACE/poly}$.
\end{theorem}

The containment $\mathsf{PSPACE/poly}\subseteq\mathsf{BPPSPACE/coin}$ is easy.
\ On the other hand, proving $\mathsf{BPPSPACE/coin}\subseteq
\mathsf{PSPACE/poly}$ appears to be no easier than the corresponding quantum
class containment. \ To prove that $\mathsf{BQPSPACE/coin}\subseteq
\mathsf{PSPACE/poly}$, we will need to understand the behavior of a
space-bounded advice coin machine $M$, as we \textit{vary} the coin bias $p$.
\ By applying a theorem of Aaronson and Watrous~\cite{awat} (which was
originally developed to understand quantum computing with closed timelike
curves), we prove the key property that, for each input $x$, \textit{the
acceptance probability }$a_{x}\left(  p\right)  $\textit{ of }$M$\textit{ is a
rational function in }$p$\textit{ of degree at most }$2^{\operatorname*{poly}%
\left(  n\right)  }$\textit{.} \ It follows that $a_{x}\left(  p\right)
$\ can \textquotedblleft oscillate\textquotedblright\ between high and low
values no more than $2^{\operatorname*{poly}\left(  n\right)  }$ times as we
vary $p$. \ Using this fact, we will show how to identify the
\textquotedblleft true\textquotedblright\ bias $p^{\ast}$\ to sufficient
precision with an advice string of $\operatorname*{poly}\left(  n\right)  $
bits. \ What makes this nontrivial is that, in our case, \textquotedblleft
sufficient precision\textquotedblright\ sometimes means $\exp\left(  n\right)
$\ bits! \ In other words, the rational functions $a_{x}\left(  p\right)
$\ really\ \textit{can} be sensitive to doubly-exponentially-small changes to
$p$. \ Fortunately, we will show that this does not happen too often, and can
be dealt with when it does.

In order to manipulate coin biases to exponentially many bits of
precision---and to interpret our advice string---in polynomial space, we use
two major tools. \ The first is a space-efficient algorithm for finding roots
of univariate polynomials, developed by Neff~\cite{neff} in the 1990s. \ The
second is a lower bound from algebraic geometry, on the spacing between
consecutive roots of a polynomial with bounded integer coefficients. \ Besides
these two tools, we will also need space-efficient linear algebra algorithms
due to Borodin, Cook, and Pippenger~\cite{bcp}.

\section{Preliminaries\label{PRELIM}}

We assume familiarity with basic notions of quantum computation. \ A detailed
treatment of space-bounded quantum Turing machines was given by
Watrous~\cite{watrous:space}.

\subsection{Classical and Quantum Space Complexity\label{spaceclasses}}

In this paper, it will generally be most convenient to consider an
\textit{asymmetric model}, in which a machine $M$ can accept only by halting
and entering a special \textquotedblleft Accept\textquotedblright\ state, but
can reject simply by never accepting. \

We say that a language $L$ is in the class $\mathsf{BPPSPACE/poly}$ if there
exists a classical probabilistic $\mathsf{PSPACE}$ machine $M$, as well as a
collection $\left\{  w_{n}\right\}  _{n\geq1}$ of polynomial-size advice
strings, such that:

\begin{enumerate}
\item[(1)] If $x\in L$, then $\Pr\left[  M\left(  x,w_{n}\right)  \text{
accepts}\right]  \geq2/3$.

\item[(2)] If $x\notin L$, then $\Pr\left[  M\left(  x,w_{n}\right)  \text{
accepts}\right]  \leq1/3$.
\end{enumerate}

Note that we do not require $M$ to accept within any fixed time bound. \ So
for example, $M$ could have expected running time that is finite, yet
\textit{doubly} exponential in $n$.\newline

The class $\mathsf{BQPSPACE/poly}$\ is defined similarly to the above, except
that now $M$ is a polynomial-space \textit{quantum} machine rather than a
classical one. \ Also, we assume that $M$ has a designated accepting state,
$\left\vert \operatorname*{Accept}\right\rangle $. \ After each computational
step, the algorithm is measured to determine whether it is in the $\left\vert
\operatorname*{Accept}\right\rangle $ state, and if so, it halts.

Watrous~\cite{watrous:space} proved the following:

\begin{theorem}
[Watrous \cite{watrous:space}]\label{bqpspace_collapse}$\mathsf{BQPSPACE/poly}%
=\mathsf{BPPSPACE/poly}=\mathsf{PSPACE/poly}$.
\end{theorem}

Note that Watrous stated his result for \textit{uniform} complexity classes,
but the proof carries over to the nonuniform case without change.

\subsection{Superoperators and Linear Algebra\label{SUPER}}

We will be interested in $S$-state quantum finite automata that can include
\textit{non-unitary transformations} such as measurements. \ The state of such
an automaton need not be a \textit{pure state} (that is, a unit vector in
$\mathbb{C}^{S}$), but can in general be a \textit{mixed state} (that is, a
probability distribution over such vectors). \ Every mixed state is uniquely
represented by an $S\times S$, Hermitian, trace-$1$\ matrix $\rho$\ called the
\textit{density matrix}. \ See Nielsen and Chuang \cite{nc}\ for more about
the density matrix formalism.

One can transform a density matrix $\rho$ using a \textit{superoperator},
which is any operation of the form%
\[
\mathcal{E}\left(  \rho\right)  =\sum_{j}E_{j}\rho E_{j}^{\dag},
\]
where the\ matrices $E_{j}\in\mathbb{C}^{S\times S}$ satisfy $\sum_{j}%
E_{j}^{\dag}E_{j}=I$.\footnote{This condition is necessary and sufficient to
ensure that $\mathcal{E}(\rho)$ is a mixed state, for every mixed state $\rho
$.}

We will often find it more convenient to work with a \textquotedblleft
vectorized\textquotedblright\ representation of mixed states and
superoperators. \ Given a density matrix $\rho\in\mathbb{C}^{S\times S}$, let
$\operatorname*{vec}\left(  \rho\right)  $ be a vector in $\mathbb{C}^{S^{2}}%
$\ containing the $S^{2}$\ entries of $\rho$. \ Similarly, given a
superoperator $\mathcal{E}$, let $\operatorname*{mat}\left(  \mathcal{E}%
\right)  \in\mathbb{C}^{S^{2}\times S^{2}}$ denote the matrix that describes
the action of $\mathcal{E}$ on vectorized mixed states, i.e., that satisfies
\[
\operatorname*{mat}\left(  \mathcal{E}\right)  \cdot\operatorname*{vec}\left(
\rho\right)  =\operatorname*{vec}\left(  \mathcal{E}\left(  \rho\right)
\right)  .
\]

We will need a theorem due to Aaronson and Watrous~\cite{awat}, which gives us
constructive access to the \textit{fixed-points} of superoperators.

\begin{theorem}
[Aaronson-Watrous \cite{awat}]\label{fixpoint}Let $\mathcal{E}\left(
\rho\right)  $ be a superoperator on an $S$-dimensional system. \ Then there
exists a second superoperator $\mathcal{E}_{\operatorname*{fix}}\left(
\rho\right)  $ on the same system, such that:

\begin{enumerate}
\item[(i)] $\mathcal{E}_{\operatorname*{fix}}\left(  \rho\right)  $ is a
fixed-point of $\mathcal{E}$ for every mixed state $\rho$: that is,
$\mathcal{E}\left(  \mathcal{E}_{\operatorname*{fix}}(\rho)\right)
=\mathcal{E}(\rho)$.

\item[(ii)] Every mixed state $\rho$ that is a fixed-point of $\mathcal{E}$ is
also a fixed-point of $\mathcal{E}_{\operatorname*{fix}}$.

\item[(iii)] Given the entries of $\operatorname*{mat}\left(  \mathcal{E}%
\right)  $, the entries of $\operatorname*{mat}\left(  \mathcal{E}%
_{\operatorname*{fix}}\right)  $ can be computed in $\operatorname*{polylog}%
(S)$\ space.
\end{enumerate}
\end{theorem}

The following fact, which we call the \textquotedblleft Leaky Subspace
Lemma,\textquotedblright will play an important role in our analysis of
quantum finite automata. \ Intuitively it says that, if repeatedly applying a
linear transformation $A$ to a vector $y$ \textquotedblleft
leaks\textquotedblright\ $y$\ into the span of another vector $x$, then there
is a uniform lower bound on the rate at which the leaking happens.

\begin{lemma}
[Leaky Subspace Lemma]\label{leaky}Let $A\in\mathbb{C}^{n\times n}$\ and
$x\in\mathbb{C}^{n}$. \ Suppose that for all vectors $y$ in some compact set
$U\subset\mathbb{C}^{n}$, there exists a positive integer $k$ such that
$x^{\dag}A^{k}y\neq0$. \ Then%
\[
\inf_{y\in U}\max_{k\in\left[  n\right]  }\left\vert x^{\dag}A^{k}y\right\vert
>0.
\]

\end{lemma}

\begin{proof}
It suffices to prove the following claim: \textit{for all }$y\in U$\textit{,
there exists a }$k\in\left[  n\right]  $\textit{\ such that }$x^{\dag}%
A^{k}y\neq0$\textit{.} \ For given this claim, Lemma~\ref{leaky} follows by
the fact that $f\left(  y\right)  :=\max_{k\in\left[  n\right]  }\left\vert
x^{\dag}A^{k}y\right\vert $\ is a continuous positive function on a compact
set $U$.

We now prove the claim. \ Let $V_{t}$\ be the vector space\ spanned by
$\left\{  Ay,A^{2}y,\ldots,A^{t}y\right\}  $, let $V:=\bigcup_{t > 0 }V_{t}$,
and let $d=\dim V$. \ Then clearly $d\leq n$ and $\dim\left(  V_{t-1}\right)
\leq\dim\left(  V_{t}\right)  \leq\dim\left(  V_{t-1}\right)  +1$\ for all
$t$. \ Now suppose $\dim\left(  V_{t}\right)  =\dim\left(  V_{t-1}\right)
$\ for some $t$. \ Then it must be possible to write $A^{t}y$\ as a linear
combination of $Ay,\ldots,A^{t-1}y$:%
\[
A^{t}y=c_{1}Ay+\cdots+c_{t-1}A^{t-1}y.
\]
But this means that every \textit{higher} iterate ($A^{t+1}y$, $A^{t+2}y$,
etc.) is also expressible as a linear combination of the lower iterates: for
example,%
\[
A^{t+1}y=c_{1}A^{2}y+\cdots+c_{t-1}A^{t}y.
\]
Therefore $d=\dim\left(  V_{t-1}\right)  $. \ The conclusion is that
$\mathcal{B}:=\left\{  Ay,A^{2}y,\ldots,A^{d}y\right\}  $ is a basis for $V$.
\ But then, if there exists a positive integer $k$ such that $v^{\dag}%
A^{k}w\neq0$, then there must also be a $k\leq d$\ such that $x^{\dag}%
A^{k}y\neq0$, by the fact that $\mathcal{B}$\ is a basis. \ This proves the claim.
\end{proof}

\subsection{Coin-Flipping Finite Automata\label{qfasec}}

It will often be convenient to use the language of finite automata rather than
that of Turing machines. \ We model a coin-flipping quantum finite automaton
as a \textit{pair} of superoperators $\mathcal{E}_{0},\mathcal{E}_{1}$. \ Say
that a coin has \textit{bias} $p$ if it lands heads with independent
probability $p$ every time it is flipped. \ (A coin here is just a
$0/1$-valued random variable, with \textquotedblleft heads\textquotedblright%
\ meaning a $1$ outcome.) \ Let $\$_{p}$\ denote a coin with bias $p$. \ When
the automaton is given $\$_{p}$, its state evolves according to the
superoperator
\[
\mathcal{E}_{p}:=p\mathcal{E}_{1}+\left(  1-p\right)  \mathcal{E}_{0}.
\]
In our model, the superoperators $\mathcal{E}_{0},\mathcal{E}_{1}$ both
incorporate a \textquotedblleft measurement step\textquotedblright\ in which
the automaton checks whether it is in a designated basis state $\left\vert
\operatorname*{Accept}\right\rangle $, and if so, halts and accepts.
\ Formally, this is represented by a projective measurement with observables
$\left\{  \Gamma_{\operatorname*{Acc}},I-\Gamma_{\operatorname*{Acc}}\right\}
$, where $\Gamma_{\operatorname*{Acc}}:=\left\vert \operatorname*{Accept}%
\right\rangle \left\langle \operatorname*{Accept}\right\vert $.

\subsection{Advice Coin Complexity Classes\label{coinsubsec}}

Given a Turing machine $M$, let $M\left(  x,\$_{p}\right)  $\ denote $M$ given
input $x$ together with the ability to flip $\$_{p}$\ at any time step. \ Then
$\mathsf{BPPSPACE/coin}$, or $\mathsf{BPPSPACE}$\textit{\ with an advice
coin}, is defined as the class of languages $L$\ for which there exists a
$\mathsf{PSPACE}$ machine $M$, as well as a sequence of real numbers $\left\{
p_{n}\right\}  _{n\geq1}$ with $p_{n}\in\left[  0,1\right]  $, such that for
all inputs $x\in\left\{  0,1\right\}  ^{n}$:

\begin{enumerate}
\item[(1)] If $x\in L$, then $M\left(  x,\$_{p_{n}}\right)  $\ accepts with
probability at least $2/3$ over the coin flips.

\item[(2)] If $x\notin L$, then $M\left(  x,\$_{p_{n}}\right)  $\ accepts with
probability at most $1/3$ over the coin flips.
\end{enumerate}

Note that there is no requirement for $M$ to halt after at most exponentially
many steps, or even to halt with probability $1$; also, $M$ may
\textquotedblleft reject\textquotedblright\ its input by looping forever. This
makes our main result, which bounds the computational power of advice coins, a
stronger statement. \ Also note that $M$ has no source of randomness other
than the coin $\$_{p_{n}}$. \ However, this is not a serious restriction,
since $M$ can easily use $\$_{p_{n}}$ to generate unbiased random bits if
needed, by using the \textquotedblleft von Neumann trick.\textquotedblright

Let $q\left(  n\right)  $ be a polynomial space bound. \ Then we model a
$q\left(  n\right)  $-space quantum Turing machine $M$ with an advice coin as
a $2^{q\left(  n\right)  }$-state automaton, with state space $\{|y\rangle
\}_{y\in\{0,1\}^{q(n)}}$ and initial state $\left\vert 0^{q(n)}\right\rangle
$. \ Given advice coin $\$_{p}$, the machine's state evolves according to the
superoperator $\mathcal{E}_{p}=p\mathcal{E}_{1}+\left(  1-p\right)
\mathcal{E}_{0}$, where $\mathcal{E}_{0},\mathcal{E}_{1}$ depend on $x$ and
$n$. \ The individual entries of the matrix representations of $\mathcal{E}%
_{0},\mathcal{E}_{1}$ are required to be computable in polynomial space.

The machine $M$ has a designated $\left\vert \operatorname*{Accept}%
\right\rangle $\ state. \ In vectorized notation, we let
$v_{\operatorname*{Acc}}:=\operatorname*{vec}\left(  \left\vert
\operatorname*{Accept}\right\rangle \left\langle \operatorname*{Accept}%
\right\vert \right)  $. \ Since $\left\vert \operatorname*{Accept}%
\right\rangle $ is a computational basis state, $v_{\operatorname*{Acc}}$ has
a single coordinate with value $1$ and is $0$ elsewhere. \ As in
Section~\ref{qfasec}, the machine measures after each computation step to
determine whether it is in the $\left\vert \operatorname*{Accept}\right\rangle
$ state.

We let $\rho_{t}$ denote the algorithm's state after $t$ steps, and let
$v_{t}:=\operatorname*{vec}\left(  \rho_{t}\right)  $. \ If we perform a
standard-basis measurement after $t$ steps, then the probability
$a_{x,t}\left(  p\right)  $ of seeing $\left\vert \operatorname*{Accept}%
\right\rangle $ is given by
\[
a_{x,t}\left(  p\right)  =\left\langle \operatorname*{Accept}\right\vert
\rho_{t}\left\vert \operatorname*{Accept}\right\rangle =v_{\operatorname*{Acc}%
}^{\dag}v_{t}.
\]
Note that $a_{x,t}\left(  p\right)  $\ is nondecreasing in $t$.

Let $a_{x}\left(  p\right)  :=\lim_{t\rightarrow\infty}a_{x,t}\left(
p\right)  $. \ Then $\mathsf{BQPSPACE/coin}$ is the class of languages $L$ for
which there exists a $\mathsf{BQPSPACE}$ machine $M$, as well as a sequence of
advice coin biases $\left\{  p_{n}\right\}  _{n\geq1}$, such that for all
$x\in\left\{  0,1\right\}  ^{n}$:

\begin{enumerate}
\item[(1)] If $x\in L$, then $a_{x}\left(  p_{n}\right)  \geq2/3$.

\item[(2)] If $x\notin L$, then $a_{x}\left(  p_{n}\right)  \leq1/3$.
\end{enumerate}

\subsection{The Hellman-Cover Theorem\label{HCSEC}}

In 1970 Hellman and Cover~\cite{hc}\ proved the following important result
(for convenience, we state only a special case).

\begin{theorem}
[Hellman-Cover Theorem \cite{hc}]\label{hcthm}Let $\$_{p}$ be a coin with bias
$p$, and let $M\left(  \$_{p}\right)  $\ be a probabilistic finite automaton
that takes as input an infinite sequence of independent flips of $\$_{p}$, and
can `halt and accept' or `halt and reject' at any time step. \ Let
$a_{t}\left(  p\right)  $ be the probability that $M\left(  \$_{p}\right)
$\ has accepted after $t$ coin flips, and let $a\left(  p\right)
=\lim_{t\rightarrow\infty}a_{t}\left(  p\right)  $. \ Suppose that $a\left(
p\right)  \leq1/3$\ and $a\left(  p+\varepsilon\right)  \geq2/3$, for some $p$
and $\varepsilon>0$. \ Then $M$\ must have $\Omega\left(  p\left(  1-p\right)
/\varepsilon\right)  $ states.
\end{theorem}

Let us make two remarks about Theorem \ref{hcthm}. \ First, the theorem is
easily seen to be essentially tight: for any $p$ and $\varepsilon>0$, one can
construct a finite automaton with $O\left(  p\left(  1-p\right)
/\varepsilon\right)  $\ states such that $a\left(  p + \varepsilon\right)  -
a\left(  p \right)  = \Omega(1)$. \ To do so, label the automaton's states by
integers in $\left\{  -K,\ldots,K\right\}  $, for some $K=O\left(  p\left(
1-p\right)  /\varepsilon\right)  $. \ Let the initial state be $0$. \ Whenever
a heads is seen, increment the state by $1$ with probability $1-p$\ and
otherwise do nothing; whenever a tails is seen, decrement the state by $1$
with probability $p$ and otherwise do nothing. \ If $K$ is ever reached, then
halt and accept (i.e., guess that the bias is $p+\varepsilon$); if $-K$\ is
ever reached, then halt and reject (i.e., guess that the bias is $p$).

Second, Hellman and Cover actually proved a stronger result. Suppose we
consider the relaxed model in which the finite automaton $M$ never needs to
halt, and one defines $a\left(  p\right)  $\ to be the fraction of time that
$M$ spends in a designated subset of `Accepting' states in the limit of
infinitely many coin flips (this limit exists with probability 1). Then the
lower bound $\Omega\left(  p\left(  1-p\right)  /\varepsilon\right)  $\ on the
number of states still holds. \ We will have more to say about finite automata
that \textquotedblleft accept in the limit\textquotedblright\ in Section
\ref{MODELS}.

\subsection{Facts About Polynomials\label{POLYFACT}}

We now collect some useful facts about polynomials and rational functions, and
about small-space algorithms for root-finding and linear algebra. \ First we
will need the following fact, which follows easily from L'H\^{o}pital's Rule.

\begin{proposition}
\label{lhopital}Whenever the limit exists,%
\[
\lim_{z\rightarrow0}\frac{c_{0}+c_{1}z+\cdots+c_{m}z^{m}}{d_{0}+d_{1}%
z+\cdots+d_{m}z^{m}}=\frac{c_{k}}{d_{k}},
\]
where $k$ is the smallest integer such that $d_{k}\neq0$.
\end{proposition}

The next two facts are much less elementary. \ First, we state a bound on the
\textit{minimum spacing} between zeros, for a low-degree polynomial with
integer coefficients.

\begin{theorem}
[{\cite[p. 359, Corollary 10.22]{basu}}]\label{sepbound}Let $P(x)$ be a
degree-$d$ univariate polynomial, with integer coefficients of bitlength at
most $\tau$. \ If $z,z^{\prime}\in\mathbb{C}$ are distinct roots of $P$, then%
\[
\left\vert z-z^{\prime}\right\vert \geq2^{-O\left(  d\log d+\tau d\right)  }.
\]
In particular, if $P$ is of degree at most $2^{\operatorname*{poly}\left(
n\right)  }$, and has integer coefficients with absolute values bounded by
$2^{\operatorname*{poly}\left(  n\right)  }$, then $\left\vert z-z^{\prime
}\right\vert \geq2^{-2^{\operatorname*{poly}\left(  n\right)  }}$.
\end{theorem}

We will need to locate the zeros of univariate polynomials to high precision
using a small amount of memory. \ Fortunately, a beautiful algorithm of
Neff~\cite{neff} from the 1990s\ (improved by Neff and Reif~\cite{neffreif}
and by Pan~\cite{pan}) provides exactly what we need.

\begin{theorem}
[\cite{neff, neffreif, pan}]\label{findroots}There exists an algorithm that

\begin{enumerate}
\item[(i)] Takes as input a triple $(P,i,j)$, where $P$ is a degree-$d$
univariate polynomial with rational\footnote{Neff's original algorithm assumes
polynomials with \textit{integer} coefficients; the result for rational
coefficients follows easily by clearing denominators.} coefficients whose
numerators and denominators are bounded in absolute value by $2^{m}$.

\item[(ii)] Outputs the $i^{th}$ most significant bits of the real and
imaginary parts of the binary expansion of the $j^{th}$ zero of $P$ (in some
order independent of $i$, possibly with repetitions).

\item[(iii)] Uses $O\left(  \operatorname*{polylog}\left(  d+i+m\right)
\right)  $ space.
\end{enumerate}
\end{theorem}

We will also need to invert $n\times n$\ matrices using
$\operatorname*{polylog}\left(  n\right)  $\ space. \ We can do so using an
algorithm of Borodin, Cook, and Pippenger~\cite{bcp} (which was also used for
a similar application by Aaronson and Watrous \cite{awat}).

\begin{theorem}
[{Borodin et al.\ \cite[Corollary 4.4]{bcp}}]\label{bcpthm}There exists an
algorithm that

\begin{enumerate}
\item[(i)] Takes as input an $n\times n$\ matrix $A=A\left(  p\right)  $,
whose entries are rational functions in $p$ of degree $\operatorname*{poly}%
\left(  n\right)  $, with the coefficients specified to $\operatorname*{poly}%
\left(  n\right)  $\ bits of precision.

\item[(ii)] Computes $\det\left(  A\right)  $ (and as a consequence, also the
$\left(  i,j\right)  $\ entry of $A^{-1}$ for any given coordinates $\left(
i,j\right)  $,\ assuming that $A$\ is invertible).

\item[(iii)] Uses $\operatorname*{poly}\left(  n\right)  $\ time and
$\operatorname*{polylog}\left(  n\right)  $\ space.
\end{enumerate}
\end{theorem}

Note that the algorithms of~\cite{bcp, neff, neffreif, pan} are all stated as
\textit{parallel} ($\mathsf{NC}$) algorithms. \ However, any parallel
algorithm can be converted into a space-efficient algorithm, using a standard
reduction due to Borodin~\cite{borodin}.

\section{Quantum Mechanics Nullifies the Hellman-Cover Theorem\label{NULLIFY}}

We now show that the quantum analogue of the Hellman-Cover Theorem
(Theorem~\ref{hcthm}) is false. \ Indeed, we will show that for any fixed
$\varepsilon>0$, there exists a quantum finite automaton with only $2$ states
that can distinguish a coin with bias $1/2$\ from a coin with bias
$1/2+\varepsilon$, with bounded probability of error independent of
$\varepsilon$. \ Furthermore, this automaton is even a \textit{halting}
automaton, which halts with probability $1$ and enters either an $\left\vert
\operatorname*{Accept}\right\rangle $\ or a $\left\vert \operatorname*{Reject}%
\right\rangle $\ state.

The key idea is that, in this setting, a single qubit can be used as an
\textquotedblleft analog counter,\textquotedblright\ in a way that a classical
probabilistic bit cannot. \ Admittedly, our result would fail were the qubit
subject to noise or decoherence, as it would be in a realistic physical situation.

Let $\rho_{0}$ be the designated starting state of the automaton, and let
$\rho_{1},\rho_{2},\ldots,$ be defined as $\rho_{t+1}=\mathcal{E}_{p}\rho_{t}%
$, with notation as in Section~\ref{coinsubsec}. \ Let
\[
a\left(  p\right)  :=\lim_{n\rightarrow\infty}\left\langle
\operatorname*{Accept}\right\vert \mathcal{E}_{p}^{n}\left(  \rho_{0}\right)
\left\vert \operatorname*{Accept}\right\rangle
\]
be the limiting probability of acceptance. This limit exists, as argued in
Section~\ref{coinsubsec}.%
\begin{figure}
[ptb]
\begin{center}
\includegraphics[
trim=1.575158in 4.154967in 2.436372in 0.000000in,
height=2.834in,
width=4.6925in
]%
{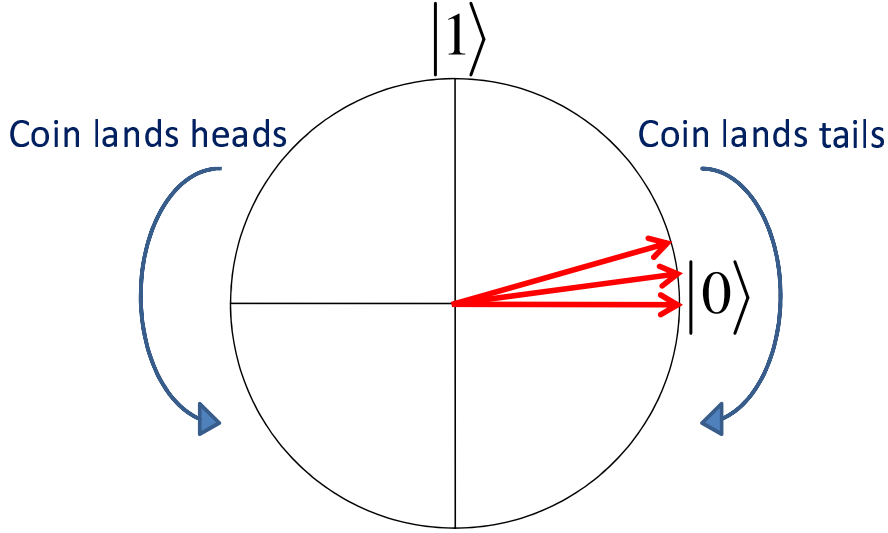}%
\caption{A quantum finite automaton that distinguishes a $p=1/2$\ coin from a
$p=1/2+\varepsilon$\ coin, essentially by using a qubit as an analog counter.}%
\label{qfafig}%
\end{center}
\end{figure}

We now prove Theorem \ref{qmnhc}, which we restate for convenience.

\begin{quotation}
\noindent\textit{Fix }$p\in\left[  0,1\right]  $\textit{ and }$\varepsilon
>0$\textit{. \ Then there exists a quantum finite automaton }$M$\textit{ with
two\ states (not counting the }$\left\vert \operatorname*{Accept}\right\rangle
$\textit{\ and }$\left\vert \operatorname*{Reject}\right\rangle $%
\textit{\ states), such that }$a\left(  p+\varepsilon\right)  -a\left(
p\right)  \geq\beta$\textit{ for some constant }$\beta$\textit{\ independent
of }$\varepsilon$\textit{. \ (For example, }$\beta=0.0117$\textit{\ works.)}
\end{quotation}

\begin{proof}
[Proof of Theorem \ref{qmnhc}]The state of $M$ will belong to the Hilbert
space spanned by $\left\{  \left\vert 0\right\rangle ,\left\vert
1\right\rangle ,\left\vert \operatorname*{Accept}\right\rangle ,\left\vert
\operatorname*{Reject}\right\rangle \right\}  $. \ The initial state is
$\left\vert 0\right\rangle $. \ Let%
\[
U\left(  \theta\right)  :=\left(
\begin{array}
[c]{cc}%
\cos\theta & -\sin\theta\\
\sin\theta & \cos\theta
\end{array}
\right)
\]
be a unitary transformation that rotates counterclockwise by $\theta$, in the
\textquotedblleft counter subspace\textquotedblright\ spanned by $\left\vert
0\right\rangle $\ and $\left\vert 1\right\rangle $. \ Also, let $A$\ and $B$
be positive integers to be specified later. \ Then the finite automaton $M$
runs the following procedure:

\begin{enumerate}
\item[(1)] If a $1$ bit is encountered (i.e., the coin lands heads), apply
$U\left(  \varepsilon\left(  1-p\right)  /A\right)  $.

\item[(2)] If a $0$ bit is encountered (i.e., the coin lands tails), apply
$U\left(  -\varepsilon p/A\right)  $.

\item[(3)] With probability $\alpha:=\varepsilon^{2}/B$, \textquotedblleft
measure\textquotedblright\ (that is, move all probability mass in $\left\vert
0\right\rangle $\ to $\left\vert \operatorname*{Reject}\right\rangle $\ and
all probability mass in $\left\vert 1\right\rangle $\ to $\left\vert
\operatorname*{Accept}\right\rangle $); otherwise do nothing.
\end{enumerate}

We now analyze the behavior of $M$. \ For simplicity, let us first consider
steps (1) and (2) only. \ In this case, we can think of $M$ as taking a random
walk in the space of possible angles between $\left\vert 0\right\rangle $\ and
$\left\vert 1\right\rangle $. \ In particular, after $t $ steps, $M$'s state
will have the form $\cos\theta_{t}\left\vert 0\right\rangle +\sin\theta
_{t}\left\vert 1\right\rangle $, for some angle $\theta_{t}\in\mathbb{R}$.
\ (As we follow the walk, we simply let $\theta_{t} $\ increase or decrease
without bound, rather than confining it to a range of size $2\pi$.) \ Suppose
the coin's bias is $p$. \ Then after $t $ steps,%
\[
\mathbb{E}\left[  \theta_{t}\right]  =pt\cdot\frac{\varepsilon}{A}\left(
1-p\right)  +\left(  1-p\right)  t\cdot\left(  -\frac{\varepsilon}{A}p\right)
=0.
\]
On the other hand, suppose the bias is $q=p+\varepsilon$. \ Then%
\begin{align*}
\mathbb{E}\left[  \theta_{t}\right]   &  =qt\cdot\frac{\varepsilon}{A}\left(
1-p\right)  +\left(  1-q\right)  t\cdot\left(  -\frac{\varepsilon}{A}p\right)
\\
&  =\frac{\varepsilon}{A}t\cdot\left[  q\left(  1-p\right)  -p\left(
1-q\right)  \right] \\
&  =\frac{\varepsilon^{2}t}{A}.
\end{align*}
So in particular, if $t=K/\varepsilon^{2}$ for some constant $K$, then
$\mathbb{E}\left[  \theta_{t}\right]  =K/A$. However, we also need to
understand the variance of the angle, $\operatorname*{Var}\left[  \theta
_{t}\right]  $. \ If the bias is $p$, then by the independence of the coin
flips,%
\begin{align*}
\operatorname*{Var}\left[  \theta_{t}\right]   &  =t\cdot\operatorname*{Var}%
\left[  \theta_{1}\right] \\
&  =t\cdot\left[  p\left(  \frac{\varepsilon}{A}\left(  1-p\right)  \right)
^{2}+\left(  1-p\right)  \left(  \frac{\varepsilon}{A}p\right)  ^{2}\right] \\
&  \leq\frac{\varepsilon^{2}t}{A^{2}},
\end{align*}
and likewise if the bias is $q=p+\varepsilon$. \ If $t=K/\varepsilon^{2}$,
this implies that $\operatorname*{Var}\left[  \theta_{t}\right]  \leq K/A^{2}%
$\ in both cases. We now incorporate step (3). \ Let $T$\ be the number of
steps before $M$\ halts (that is, before its state gets measured). \ Then
clearly $\Pr\left[  T=t\right]  =\alpha\left(  1-\alpha\right)  ^{t}$. \ Also,
let $u:=K/\varepsilon^{2}$ for some $K$ to be specified later. \ Then if the
bias is $p$, we can upper-bound $M$'s acceptance probability $a\left(
p\right)  $\ as%
\begin{align*}
a\left(  p\right)   &  =\sum_{t=1}^{\infty}\Pr\left[  T=t\right]
\cdot\mathbb{E}\left[  \sin^{2}\theta_{t}~|~t\right] \\
&  \leq\Pr\left[  T>u\right]  +\sum_{t=1}^{u}\Pr\left[  T=t\right]
\cdot\mathbb{E}\left[  \sin^{2}\theta_{t}~|~t\right] \\
&  \leq\Pr\left[  T>u\right]  +\sum_{t=1}^{u}\Pr\left[  T=t\right]
\cdot\mathbb{E}\left[  \theta_{t}^{2}~|~t\right] \\
&  \leq\left(  1-\alpha\right)  ^{u}+\mathbb{E}\left[  \theta_{u}%
^{2}~|~u\right] \\
&  \leq\left(  1-\frac{\varepsilon^{2}}{B}\right)  ^{B/\varepsilon^{2}\cdot
K/B}+\frac{\varepsilon^{2}u}{A^{2}}\\
&  \leq e^{-K/B}+\frac{K}{A^{2}}.
\end{align*}
Here the third line uses $\sin x\leq x$, while the fourth line uses the fact
that $\mathbb{E}\left[  \theta_{t}^{2}\right]  $\ is nondecreasing for an
unbiased random walk. \ So long as $A^{2}\geq B$, we can minimize the final
expression by setting $K:=B\ln\left(  A^{2}/B\right)  $, in which case we have%
\[
a\left(  p\right)  \leq\frac{B}{A^{2}}\left(  1+\ln\left(  \frac{A^{2}}{B}
\right)  \right)  .
\]
On the other hand, suppose the bias is $p+\varepsilon$. \ Set
$v:=L/\varepsilon^{2}$ where $L:=\pi A/4$.\ \ Then for all $t\leq v$, we have%
\begin{align*}
\Pr\left[  \left\vert \theta_{t}\right\vert >\pi/2~|~t\right]   &  \leq
\Pr\left[  \left\vert \theta_{t}-\frac{\varepsilon^{2}t}{A}\right\vert
>\frac{\pi}{2}-\frac{\varepsilon^{2}t}{A}~|~t\right] \\
&  <\frac{\varepsilon^{2}t/A^{2}}{\left(  \pi/2-\varepsilon^{2}t/A\right)
^{2}}\\
&  \leq\frac{\varepsilon^{2}v/A^{2}}{\left(  \pi/2-\varepsilon^{2}v/A\right)
^{2}}\\
&  =\frac{4}{\pi A}%
\end{align*}
where the second line uses Chebyshev's inequality. \ Also, let $\Delta
_{t}:=\theta_{t}-\varepsilon^{2}t/A$. \ Then for all $t\leq v$\ we have%
\begin{align*}
\mathbb{E}\left[  \theta_{t}^{2}~|~t\right]   &  =\mathbb{E}\left[  \left(
\frac{\varepsilon^{2}t}{A}+\Delta_{t}\right)  ^{2}~|~t\right] \\
&  =\frac{\varepsilon^{4}t^{2}}{A^{2}}+\mathbb{E}\left[  \Delta_{t}%
^{2}~|~t\right]  +2\frac{\varepsilon^{2}t}{A}\mathbb{E}\left[  \Delta
_{t}~|~t\right] \\
&  \geq\frac{\varepsilon^{4}t^{2}}{A^{2}}.
\end{align*}
Putting the pieces together, we can lower-bound $a\left(  p+\varepsilon
\right)  $\ as%
\begin{align*}
a\left(  p+\varepsilon\right)   &  =\sum_{t=1}^{\infty}\Pr\left[  T=t\right]
\cdot\mathbb{E}\left[  \sin^{2}\theta_{t}~|~t\right] \\
&  \geq\sum_{t=1}^{v}\Pr\left[  T=t\right]  \cdot\mathbb{E}\left[  \sin
^{2}\theta_{t}~|~t\right] \\
&  \geq\sum_{t=1}^{v}\Pr\left[  T=t\right]  \cdot\Pr\left[  \left\vert
\theta_{t}\right\vert \leq\pi/2~|~t\right]  \cdot\mathbb{E}\left[  \theta
_{t}^{2}/3~|~t\right] \\
&  \geq\sum_{t=1}^{v}\alpha\left(  1-\alpha\right)  ^{t}\cdot\left(
1-\frac{4}{\pi A}\right)  \cdot\frac{\varepsilon^{4}t^{2}}{3A^{2}}\\
&  =\left(  1-\frac{4}{\pi A}\right)  \frac{\varepsilon^{4}\alpha}{3A^{2}}%
\sum_{t=1}^{L/\varepsilon^{2}}\left(  1-\frac{\varepsilon^{2}}{B}\right)
^{t}t^{2}\\
&  \geq\left(  1-\frac{4}{\pi A}\right)  \frac{\varepsilon^{4}\alpha}%
{3A^{2}e^{L/B}}\sum_{t=1}^{L/\varepsilon^{2}}t^{2}\\
&  \geq\left(  1-\frac{4}{\pi A}\right)  \frac{\varepsilon^{6}}{3A^{2}%
Be^{L/B}}\cdot\frac{\left(  L/\varepsilon^{2}\right)  ^{3}}{6}\\
&  =\left(  1-\frac{4}{\pi A}\right)  \frac{L^{3}}{18A^{2}Be^{L/B}}\\
&  =\left(  1-\frac{4}{\pi A}\right)  \frac{\pi^{3}A}{1152Be^{\pi A/4B}}.
\end{align*}
Here the third line uses the fact that $\sin^{2}x\geq x^{2}/3$ for all
$\left\vert x\right\vert \leq\pi/2$. \ If we now choose (for example) $A=10000
$\ and $B=7500$, then we have $a\left(  p\right)  \leq0.0008$ and $a\left(
p+\varepsilon\right)  \geq0.0125$, whence $a\left(  p+\varepsilon\right)
-a\left(  p\right)  \geq0.0117.$
\end{proof}

We can strengthen Theorem \ref{qmnhc}\ to ensure that $a\left(  p\right)
\leq\delta$\ and $a\left(  p+\varepsilon\right)  \geq1-\delta$ for any desired
error probability $\delta>0$. \ We simply use standard amplification, which
increases the number of states in $M$ to $O\left(  \operatorname*{poly}\left(
1/\delta\right)  \right)  $\ (or equivalently, the number of qubits to
$O\left(  \log(1/\delta)\right)  $).

\section{Upper-Bounding the Power of Advice Coins\label{MAINRESULT}}

In this section we prove Theorem~\ref{mainthm}, that $\mathsf{BQPSPACE/coin}%
=\mathsf{BPPSPACE/coin}=\mathsf{PSPACE/poly}$. \ We start with the easy half:

\begin{proposition}
$\mathsf{PSPACE/poly}\subseteq\mathsf{BPPSPACE/coin} \subseteq
\mathsf{BQPSPACE/coin}$.
\end{proposition}

\begin{proof}
Given a polynomial-size advice string $w_{n}\in\left\{  0,1\right\}
^{s\left(  n\right)  }$, we encode $w_{n}$\ into the first $s\left(  n\right)
$\ bits of the binary expansion of an advice bias $p_{n}\in\left[  0,1\right]
$. \ Then by flipping the coin $\$_{p_{n}}$\ sufficiently many times
($O\left(  2^{2s\left(  n\right)  }\right)  $ trials suffice) and tallying the
fraction of heads, a Turing machine can recover $w_{n}$\ with high success
probability. \ Counting out the desired number of trials and determining the
fraction of heads seen can be done in space $O\left(  \log\left(  2^{2s\left(
n\right)  }\right)  \right)  =O\left(  s\left(  n\right)  \right)  =O\left(
\operatorname*{poly}\left(  n\right)  \right)  $. \ Thus we can simulate a
$\mathsf{PSPACE/poly}$ machine with a $\mathsf{BPPSPACE/coin}$ machine.
\end{proof}

The rest of the section is devoted to showing that $\mathsf{BQPSPACE/coin}%
\subseteq\mathsf{PSPACE/poly}$. First we give some lemmas about quantum
polynomial-space advice coin algorithms. Let $M$ be such an algorithm.
\ Suppose $M$ uses $s\left(  n\right)  =\operatorname*{poly}\left(  n\right)
$\ qubits of memory, and has $S=2^{s\left(  n\right)  }$\ states. \ Let
$\mathcal{E}_{0},\mathcal{E}_{1},\mathcal{E}_{p}$ be the superoperators for
$M$ as described in Section~\ref{coinsubsec}. \ Recalling the vectorized
notation from Section~\ref{coinsubsec}, let $B_{p}:=\operatorname*{mat}\left(
\mathcal{E}_{p}\right)  $. \ Let $\rho_{x,t}\left(  p\right)  $ be the state
of $M$ after $t$ coin flips steps on input $x$ and coin bias $p$, and let
$v_{x,t}\left(  p\right)  :=\operatorname*{vec}\left(  \rho_{x,t}\left(
p\right)  \right)  $. \ Let%
\[
a_{x,t}\left(  p\right)  :=v_{\operatorname*{Acc}}^{\dag}v_{x,t}\left(
p\right)
\]
be the probability that $M$ is in the $\left\vert \operatorname*{Accept}%
\right\rangle $ state, if measured after $t$ steps. \ Let $a_{x}\left(
p\right)  :=\lim_{t\rightarrow\infty}a_{x,t}\left(  p\right)  $. \ As
discussed in Section~\ref{coinsubsec}, the quantities $a_{x,t}\left(
p\right)  $ are nondecreasing in $t$, so the limit $a_{x}\left(  p\right)  $
is well-defined.

We now show that---except possibly at a finite number of values---$a_{x}%
\left(  p\right)  $\ is actually a \textit{rational function} of $p$, whose
degree is at most the number of states.

\begin{lemma}
\label{rational}There exist polynomials $Q(p)$ and $R(p)\neq0$, of degree at
most $S^{2}=2^{\operatorname*{poly}\left(  n\right)  }$ in $p$, such that
\[
a_{x}\left(  p\right)  =\frac{Q\left(  p\right)  }{R\left(  p\right)  }%
\]
holds whenever $R\left(  p\right)  \neq0$. \ Moreover, $Q$ and $R$ have
rational coefficients that are computable in $\operatorname*{poly}\left(
n\right)  $ space given $x\in\{0,1\}^{n}$ and the index of the desired coefficient.
\end{lemma}

\begin{proof}
Throughout, we suppress the dependence on $x$ for convenience, so that
$a\left(  p\right)  =\lim_{t\rightarrow\infty}a_{t}\left(  p\right)  $\ is
simply the limiting acceptance probability of a finite automaton $M\left(
\$_{p}\right)  $ given a coin with bias $p$.

Following Aaronson and Watrous \cite{awat}, for $z\in\left(  0,1\right)  $
define the matrix $\Lambda_{z,p}\in\mathbb{C}^{S^{2}\times S^{2}}$ by
\[
\Lambda_{z,p}:=z\left[  I-\left(  1-z\right)  B_{p}\right]  ^{-1}.
\]
The matrix $I-\left(  1-z\right)  B_{p}$ is invertible, since $z>0$\ and all
eigenvalues of $B_{p}$\ have absolute value at most $1$.\footnote{For the
latter fact, see~\cite{td:equilib} and~\cite[p. 10, footnote 1]{awat}.}
\ Using Cramer's rule, we can represent each entry of $\Lambda_{z,p}$\ in the
form $\frac{f\left(  z,p\right)  }{g\left(  z,p\right)  }$, where $f$ and $g$
are bivariate polynomials of degree at most $S^{2}$ in both $z$ and $p$, and
$g\left(  z,p\right)  $\ is not identically zero. \ Note that by collecting
terms, we can write%
\begin{align*}
f\left(  z,p\right)   &  =c_{0}\left(  p\right)  +c_{1}\left(  p\right)
z+\cdots+c_{S^{2}}\left(  p\right)  z^{S^{2}}\\
g\left(  z,p\right)   &  =d_{0}\left(  p\right)  +d_{1}\left(  p\right)
z+\cdots+d_{S^{2}}\left(  p\right)  z^{S^{2}},
\end{align*}
for some coefficients $c_{0},\ldots,c_{S^{2}}$\ and $d_{0},\ldots,d_{S^{2}}$.
Now let%
\begin{equation}
\Lambda_{p}:=\lim_{z\rightarrow0}\Lambda_{z,p}. \label{phi_eq}%
\end{equation}
Aaronson and Watrous~\cite{awat} showed that $\Lambda_{p}$ is precisely the
matrix representation $\operatorname*{mat}\left(  \mathcal{E}%
_{\operatorname*{fix}}\right)  $ of the superoperator $\mathcal{E}%
_{\operatorname*{fix}}$ associated to $\mathcal{E}:=\mathcal{E}_{p}$ by
Theorem~\ref{fixpoint}. \ Thus we have
\[
B_{p}\left(  \Lambda_{p}v\right)  =\Lambda_{p}v
\]
for all $v\in\mathbb{C}^{S^{2}}$.

Now, the entries of $\Lambda_{z,p}$\ are bivariate rational functions, which
have absolute value at most $1$ for all $z,p$. \ Thus the limit in
equation~(\ref{phi_eq}) must exist, and the coeffients $c_{k}$, $d_{k}$ can be
computed in polynomial space using Theorem~\ref{bcpthm}.

We claim that every entry of $\Lambda_{p}$ can be represented as a rational
function of $p$ of degree at most $S^{2}$ (a representation valid for all but
finitely many $p$), and that the coefficients of this rational function are
computable in polynomial space. To see this, fix some $i,j\in\left[  S\right]
$, and let $\left(  \Lambda_{p}\right)  _{ij}$\ denote the $\left(
i,j\right)  ^{th}$\ entry of $\Lambda_{p}$. \ By the above, $\left(
\Lambda_{p}\right)  _{ij}$\ has the form%
\[
\left(  \Lambda_{p}\right)  _{ij}=\lim_{z\rightarrow0}\frac{f\left(
z,p\right)  }{g\left(  z,p\right)  }=\lim_{z\rightarrow0}\frac{c_{0}\left(
p\right)  +c_{1}\left(  p\right)  z+\cdots+c_{S^{2}}\left(  p\right)
z^{S^{2}}}{d_{0}\left(  p\right)  +d_{1}\left(  p\right)  z+\cdots+d_{S^{2}%
}\left(  p\right)  z^{S^{2}}}.
\]
By Proposition \ref{lhopital}, the above limit (whenever it exists) equals
$c_{k}\left(  p\right)  /d_{k}\left(  p\right)  $, where $k$ is the smallest
integer such that $d_{k}\left(  p\right)  \neq0$. \ Now let $k^{\ast}$\ be the
smallest integer such that $d_{k^{\ast}}$\ is not the identically-zero
polynomial. \ Then $d_{k^{\ast}}\left(  p\right)  $\ has only finitely many
zeros. It follows that $\left(  \Lambda_{p}\right)  _{ij}=c_{k^{\ast}}\left(
p\right)  /d_{k^{\ast}}\left(  p\right)  $ except when $d_{k^{\ast}}\left(
p\right)  =0$, which is what we wanted to show. \ That the coefficients are
rational and computable in polynomial space follows by construction: we can
loop through all $k$\ until we find $k^{\ast}$\ as above, and then compute the
coefficients of $c_{k^{\ast}}\left(  p\right)  $\ and $d_{k^{\ast}}\left(
p\right)  $.

Finally, we claim that we can write $A$'s limiting acceptance probability
$a\left(  p\right)  $ as%
\begin{equation}
a\left(  p\right)  =v_{\operatorname*{Acc}}^{\dag}\Lambda_{p}v_{0}%
,\label{accepteq}%
\end{equation}
where $v_{0}$ is the vectorized initial state of $A$ (independent of $p$). It
will follow from equation~(\ref{accepteq}) that $a\left(  p\right)  $ has the
desired rational-function representation, since the map $\Lambda
_{p}\rightarrow v_{\text{Acc}}^{\dag}\Lambda_{p}v_{0}$ is linear in the
entries of $\Lambda_{p}$ and can be performed in polynomial space.

To establish equation~(\ref{accepteq}), consider the Taylor series expansion
for $\Lambda_{z,p}$,
\[
\Lambda_{z,p}=\sum_{t\geq0}z(1-z)^{t}B_{p}^{t},
\]
valid for $z\in(0,1)$ (see~\cite{awat} for details). \ The equality
\[
\sum_{t\geq0}z(1-z)^{t}=1,\quad{}z\in(0,1),
\]
implies that $v_{\operatorname*{Acc}}^{\dag}\Lambda_{z,p}v_{0}$ is a weighted
average of the $t$-step acceptance probabilities $a_{t}\left(  p\right)  $,
for $t\in\{0,1,2,\ldots\}$. Letting $z\rightarrow0$, the weight on each
individual step approaches 0. Since $\lim_{t\rightarrow\infty}a_{t}(p)=a(p)$,
we obtain equation~(\ref{accepteq}).
\end{proof}

The next lemma lets us \textquotedblleft patch up\textquotedblright\ the
finitely many singularities, and show that $a_{x}\left(  p\right)  $\ is a
rational function in the entire open interval $\left(  0,1\right)
$.\footnote{Note that there could still be singularities at $p=0$\ and $p=1$,
and this is not just an artifact of the proof! \ For example, consider a
finite automaton that accepts when and only when it sees `heads.' \ The
acceptance probability of such an automaton satisfies $a\left(  0\right)  =0$,
but $a\left(  p\right)  =1$ for all $p\in\left(  0,1\right]  $.}

\begin{lemma}
\label{continuousq}$a_{x}\left(  p\right)  $ is continuous for all
$p\in\left(  0,1\right)  $.
\end{lemma}

\begin{proof}
Once again we suppress the dependence on $x$, so that $a\left(  p\right)
=\lim_{t\rightarrow\infty}a_{t}\left(  p\right)  $\ is just the limiting
acceptance probability of a finite automaton $M\left(  \$_{p}\right)  $.

To show that $a\left(  p\right)  $ is continuous on $\left(  0,1\right)  $, it
suffices to show that $a\left(  p\right)  $\ is continuous on every closed
subinterval $\left[  p_{1},p_{2}\right]  $ such that $0<p_{1}<p_{2}<1$. \ We
will prove this by proving the following claim:

\begin{quotation} \noindent \textit{(*) For every subinterval }$\left[  p_{1},p_{2}\right]  $\textit{\ and
every }$\delta>0$\textit{, there exists a time }$t$\textit{ (not depending on
}$p$\textit{) such that }$a_{t}\left(  p\right)  \geq a\left(  p\right)
-\delta$\textit{\ for all }$p\in\left[  p_{1},p_{2}\right]  $. \end{quotation}

Claim (*) implies that $a\left(  p\right)  $\ can be uniformly approximated by
continuous functions on $\left[  p_{1},p_{2}\right]  $, and hence is
continuous itself on $\left[  p_{1},p_{2}\right]  $.

We now prove claim (*). \ First, call a mixed state $\rho$\ \textit{dead for
bias} $p$ if $M\left(  \$_{p}\right)  $ halts with probability $0$ when run
with $\rho$\ as its initial state. \ Now, the superoperator applied by
$M\left(  \$_{p}\right)  $\ at each time step is $\mathcal{E}_{p}%
=p\mathcal{E}_{1}+\left(  1-p\right)  \mathcal{E}_{0}$. \ This means that
$\rho$ is dead for any bias $p\in\left(  0,1\right)  $, if and only if $\rho$
is dead for bias $p=1/2$. \ So we can simply refer to such a $\rho$ as
\textit{dead}, with no dependence on $p$.

Recall that $B_{p}:=\operatorname*{mat}\left(  \mathcal{E}_{p}\right)  $.
\ Observe that $\rho$ is dead if and only if%
\[
v_{\operatorname*{Acc}}^{\dag}B_{1/2}^{t}\operatorname*{vec}\left(
\rho\right)  =0
\]
for all $t\geq0$. \ In particular, it follows that there exists a
\textquotedblleft dead subspace\textquotedblright\ $D$\ of $\mathbb{C}^{S}$,
such that a pure state $\left\vert \psi\right\rangle $\ is dead if and only if
$\left\vert \psi\right\rangle \in D$. \ (A mixed state $\rho=\sum_{i}%
p_{i}\left\vert \psi_{i}\right\rangle \left\langle \psi_{i}\right\vert $ is
dead if and only if $\left\vert \psi_{i}\right\rangle \in D$ for all $i$ such
that $p_{i}>0$.) \ By its definition, $D$\ is orthogonal to the $\left\vert
\operatorname*{Accept}\right\rangle $ state. \ Define the \textquotedblleft
live subspace,\textquotedblright\ $L$, to be the orthogonal complement of
$\left\vert \operatorname*{Accept}\right\rangle $ and $D$.

Let $P$ be the projector onto $L$, and let $v_{\operatorname*{Live}%
}:=\operatorname*{vec}\left(  P\right)  $. \ Also, recalling that $v_{0}$\ is
the vectorized initial state of $M$, let%
\[
g_{t}\left(  p\right)  :=v_{\operatorname*{Live}}^{\dag}B_{p}^{t}v_{0}%
\]
be the probability that $M\left(  \$_{p}\right)  $\ is \textquotedblleft still
alive\textquotedblright\ if measured after $t$ steps---i.e., that $M$ has
neither accepted nor entered the dead subspace. \ Clearly $a\left(  p\right)
\leq a_{t}\left(  p\right)  +g_{t}\left(  p\right)  $.

Thus, to prove claim (*), it suffices to prove that for all $\delta>0$, there
exists a $t$ (not depending on $p$) such that $g_{t}\left(  p\right)
\leq\delta$\ for all $p\in\left[  p_{1},p_{2}\right]  $. \ First, let $U$\ be
the set of all $\rho$\ supported only on the live subspace $L$, and notice
that $U$ is compact. \ Therefore, by Lemma \ref{leaky}\ (the \textquotedblleft
Leaky Subspace Lemma\textquotedblright), there exists a constant $c_{1}%
>0$\ such that, for all $\rho\in U$,%
\[
\left(  v_{\operatorname*{Acc}}^{\dag}+v_{\operatorname*{Dead}}^{\dag}\right)
B_{p_{1}}^{S^{2}}\operatorname*{vec}\left(  \rho\right)  \geq c_{1}%
\]
and hence%
\[
v_{\operatorname*{Live}}^{\dag}B_{p_{1}}^{S^{2}}\operatorname*{vec}\left(
\rho\right)  \leq1-c_{1}.
\]
Likewise, there exists a $c_{2}>0$\ such that, for all $\rho\in U$,%
\[
v_{\operatorname*{Live}}^{\dag}B_{p_{2}}^{S^{2}}\operatorname*{vec}\left(
\rho\right)  \leq1-c_{2}.
\]
Let $c:=\min\left\{  c_{1},c_{2}\right\}  $. \ Then by convexity, for all
$p\in\left[  p_{1},p_{2}\right]  $\ and all $\rho\in U$, we have%
\[
v_{\operatorname*{Live}}^{\dag}B_{p}^{S^{2}}\operatorname*{vec}\left(
\rho\right)  \leq1-c,
\]
and hence%
\[
v_{\operatorname*{Live}}^{\dag}B_{p}^{S^{2}t}\operatorname*{vec}\left(
\rho\right)  \leq\left(  1-c\right)  ^{t}%
\]
for all $t\geq0$. \ This means that, to ensure that $g_{t}\left(  p\right)
\leq\delta$\ for all $p\in\left[  p_{1},p_{2}\right]  $ simultaneously, we
just need to choose $t$ large enough that $\left(  1-c\right)  ^{t/S^{2}}%
\leq\delta$. \ This proves claim (*).
\end{proof}

We are now ready to complete the proof of Theorem~\ref{mainthm}. \ Let $L$ be
a language in $\mathsf{BQPSPACE/coin}$, which is decided by the quantum
polynomial-space advice-coin machine $M\left(  x,\$_{p}\right)  $ on advice
coin biases $\left\{  p_{n}\right\}  _{n\geq1}$. We will show that $L
\in\mathsf{BQPSPACE/poly} = \mathsf{PSPACE/poly}$.

It may not be possible to perfectly specify the bias $p_{n}$ using
$\operatorname*{poly}\left(  n\right)  $ bits of advice. Instead, we use our
advice string to simulate access to a second bias $r_{n}$ that is
\textquotedblleft almost as good\textquotedblright\ as $p_{n}$. This is
achieved by the following lemma.

\begin{lemma}
\label{findbias}Fixing $L,M,\{p_{n}\}$ as above, there exists a classical
polynomial-space algorithm $R$, as well as a family $\left\{  w_{n}\right\}
_{n\geq1}$ of polynomial-size advice strings, for which the following holds.
\ Given an index $i\leq2^{\operatorname*{poly}\left(  n\right)  }$, the
computation $R\left(  w_{n},i\right)  $ outputs the $i^{th}$ bit of a real
number $r_{n}\in\left(  0,1\right)  $,\ such that for all $x\in\left\{
0,1\right\}  ^{n}$,

\begin{enumerate}
\item[(i)] If $x\in L$, then $\Pr\left[  M\left(  x,\$_{r_{n}}\right)  \text{
accepts}\right]  \geq3/5$.

\item[(ii)] If $x\notin L$, then $\Pr\left[  M\left(  x,\$_{r_{n}}\right)
\text{ accepts}\right]  \leq2/5$.

\item[(iii)] The binary expansion of $r_{n}$ is identically zero, for
sufficiently large indices $j\geq h(n)=2^{\operatorname*{poly}\left(
n\right)  }$.
\end{enumerate}
\end{lemma}

Once Lemma~\ref{findbias} is proved, showing the containment $L\in
\mathsf{BQPSPACE/poly}$ is easy. \ First, we claim that using the advice
family $\left\{  w_{n}\right\}  $, we can simulate access to $r_{n}$-biased
coin flips, as follows. \ Let $r_{n}=0.b_{1}b_{2}\ldots$ denote the binary
expansion of $r_{n}$.\bigskip

\texttt{given }$w_{n}$

$j := 0$

\texttt{while }$j<h\left(  n\right)  $

\qquad\texttt{let }$z_{j}\in\left\{  0,1\right\}  $\texttt{ be random}

\qquad$b_{j}:=R\left(  w_{n},j\right)  $

\qquad\texttt{if }$z_{j}<b_{j}$\texttt{\ then output }1

\qquad\texttt{else if }$z_{j}>b_{j}$\texttt{\ then output }$\mathtt{0}$

\qquad\texttt{else }$j := j+1$

\texttt{output }$0$\bigskip

Observe that this algorithm, which runs in polynomial space, outputs $1$ if
and only if
\[
0.z_{1}z_{2}\ldots z_{h\left(  n\right)  }<0.b_{1}b_{2}\ldots b_{h\left(
n\right)  }=r_{n},
\]
and this occurs with probability $r_{n}$. Thus we can simulate $r_{n}$-biased
coin flips as claimed.

We define a $\mathsf{BQPSPACE/poly}$ machine $M^{\prime}$ that takes $\left\{
w_{n}\right\}  $ from Lemma~\ref{findbias} as its advice. \ Given an input
$x\in\left\{  0,1\right\}  ^{n}$, the machine $M^{\prime}$ simulates $M\left(
x,\$_{r_{n}}\right)  $, by generating $r_{n}$-biased coin flips using the
method described above. Then $M^{\prime}$ is a $\mathsf{BQPSPACE/poly}$
algorithm for $L$ by parts (i) and (ii) of Lemma~\ref{findbias}, albeit with
error bounds $\left(  2/5,3/5\right)  $. \ The error bounds can be boosted to
$\left(  1/3,2/3\right)  $ by running several independent trials. \ So
$L\in\mathsf{BQPSPACE/poly}=\mathsf{PSPACE/poly}$, completing the proof of
Theorem~\ref{mainthm}.

\begin{proof}
[Proof of Lemma~\ref{findbias}]Fix an input length $n>0$, and let $p^{\ast
}:=p_{n}$. \ For $x\in\{0,1\}^{n}$, recall that $a_{x}\left(  p\right)  $
denotes the acceptance probability of $M\left(  x,\$_{p}\right)  $. \ We are
interested in the way $a_{x}\left(  p\right)  $ oscillates as we vary $p$.
\ Define a \textit{transition pair} to be an ordered pair $\left(  x,p\right)
\in\left\{  0,1\right\}  ^{n}\times\left(  0,1\right)  $\ such that
$a_{x}\left(  p\right)  \in\left\{  2/5,3/5\right\}  $. \ It will be also be
convenient to define a larger set of \textit{potential transition pairs},
denoted $\mathcal{P}\subseteq\left\{  0,1\right\}  ^{n}\times\left[
0,1\right)  $, that contains the transition pairs; the benefit of considering
this larger set is that its elements will be easier to enumerate. \ We defer
the precise definition of $\mathcal{P}$.

The advice string $w_{n}$ will simply specify \textit{the number of distinct
potential transition pairs} $\left(  y,p\right)  $ \textit{such that} $p\leq
p^{\ast}$. \ We first give a high-level pseudocode description of the
algorithm $R$; after proving that parts (i) and (ii) of Lemma~\ref{findbias}
are met by the algorithm, we will fill in the algorithmic details to show that
the pseudocode can be implemented in $\mathsf{PSPACE}$, and that we can
satisfy part (iii) of the Lemma.

The pseudocode for $R$ is as follows:\bigskip

\texttt{given }$\left(  w_{n},i\right)  $

\texttt{for all }$\left(  y,p\right)  \in\mathcal{P}$

$\qquad s:=0$

$\qquad$\texttt{for all }$\left(  z,q\right)  \in\mathcal{P}$

\qquad\qquad\texttt{if }$q \leq p$\texttt{\ then }$s:=s+1$

\qquad\texttt{next\ }$\left(  z,q\right)  $

\qquad\texttt{if }$s=w_{n}$\texttt{\ then}

\qquad\qquad\texttt{let }$r_{n} :=p+\varepsilon$\texttt{ (for some small
}$\varepsilon=2^{-2^{\operatorname*{poly}\left(  n\right)  }}$\texttt{)}

\qquad\qquad\texttt{output the $i$}$^{\mathtt{th}}$\texttt{ bit of }$r_{n}$

\qquad\texttt{end if}

\texttt{next }$\left(  y,p\right)  \bigskip$

We now prove that parts (i) and (ii) of Lemma~\ref{findbias} are satisfied.
\ We call $p\in\left[  0,1\right)  $ a \textit{transition value} if $\left(
y,p\right)  $ is a transition pair for some $y\in\left\{  0,1\right\}  ^{n}$,
and we call $p$ a \textit{potential transition value} if $\left(  y,p\right)
\in\mathcal{P}$ for some $y\in\left\{  0,1\right\}  ^{n}$. \ Then by
definition of $w_{n}$, the value $r_{n}$ produced above is equal to
$p_{0}+\varepsilon$, where $p_{0}\in\lbrack0,1)$ is the largest potential
transition value less than or equal to $p^{\ast}$. \ (Note that $0$ will
always be a potential transition value, so this is well-defined.)

When we define $\mathcal{P}$, we will argue that any distinct potential
transition values $p_{1},p_{2}$ satisfy%
\begin{equation}
\min\left\{  \left\vert p_{1}-p_{2}\right\vert ,1-p_{2}\right\}
\geq2^{-2^{\operatorname*{poly}\left(  n\right)  }}. \label{ptp_sep}%
\end{equation}
It follows that if $\varepsilon=2^{-2^{\operatorname*{poly}\left(  n\right)
}}$ is suitably small, then $r_{n}<1$, and there is no potential transition
value lying in the range $\left(  p_{0},r_{n}\right]  $. \ Also, there are no
potential transition values in the interval $\left(  p_{0},p^{\ast}\right)  $.

Now fix any $x\in\left\{  0,1\right\}  ^{n}\cap L$. \ Since $M$ is a
$\mathsf{BQPSPACE/coin}$ machine for $L$ with bias $p^{\ast}$, we have
$a_{x}\left(  p^{\ast}\right)  \geq2/3$. \ If $a_{x}\left(  r_{n}\right)
<3/5$, then Lemma~\ref{continuousq} implies that there must be a transition
value in the open interval between $p^{\ast}$ and $r_{n}$. \ But there are no
such transition values. \ Thus $a_{x}\left(  r_{n}\right)  \geq3/5$.
\ Similarly, if $x\in\left\{  0,1\right\}  ^{n}\setminus L$, then
$a_{x}\left(  r_{n}\right)  \leq2/5$. \ This establishes parts (i) and (ii) of
Lemma~\ref{findbias}.

Now we formally define the potential transition pairs $\mathcal{P}$. \ We
include $\left(  0^{n},0\right)  $ in $\mathcal{P}$, guaranteeing that $0$ is
a potential transition value as required. \ Now recall, by
Lemma~\ref{rational}, that for each $x\in\left\{  0,1\right\}  ^{n}$, the
acceptance probability $a_{x}\left(  p\right)  $ is a rational function
$Q_{x}\left(  p\right)  /R_{x}\left(  p\right)  $ of degree
$2^{\operatorname*{poly}\left(  n\right)  }$, for all but finitely many
$p\in\left(  0,1\right)  $. \ Therefore, the function $\left(  a_{x}\left(
p\right)  -3/5\right)  \left(  a_{x}\left(  p\right)  -2/5\right)  $ also has
a rational-function representation:
\[
\frac{U_{x}\left(  p\right)  }{V_{x}\left(  p\right)  }=\left(  a_{x}\left(
p\right)  -\frac{3}{5}\right)  \left(  a_{x}\left(  p\right)  -\frac{2}%
{5}\right)  ,
\]
valid for all but finitely many $p$. \ We will include in $\mathcal{P}$ all
pairs $\left(  x,p\right)  $ for which $U_{x}\left(  p\right)  =0$. \ It
follows from Lemmas~\ref{rational} and~\ref{continuousq} that $\mathcal{P}$
contains all transition pairs, as desired.

We can now establish equation~(\ref{ptp_sep}). \ Fix any distinct potential
transition values $p_{1}<p_{2}$ in $\left[  0,1\right)  $. \ Since $p_{2}%
\neq0$ is a potential transition value, there is some $x_{2}$ such that
$\left(  x_{2},p_{2}\right)  \in\mathcal{P}$. \ If $p_{1}=0$, then
$p_{1},p_{2}$ are distinct roots of the polynomial $pU_{x_{2}}(p)$, whence
$\left\vert p_{1}-p_{2}\right\vert \geq2^{-2^{\operatorname*{poly}\left(
n\right)  }}$ by Theorem~\ref{sepbound}. \ Similarly, if $p_{1}>0$, then
$\left(  x_{1},p_{1}\right)  \in\mathcal{P}$ for some $x_{1}$. \ We observe
that $p_{1},p_{2}$ are common roots of $U_{x_{1}}\left(  p\right)  U_{x_{2}%
}\left(  p\right)  $, from which it again follows that $\left\vert p_{1}%
-p_{2}\right\vert \geq2^{-2^{\operatorname*{poly}\left(  n\right)  }}$.
\ Finally, $1-p_{2}\geq2^{-2^{\operatorname*{poly}\left(  n\right)  }}$
follows since $1$ and $p_{2}$ are distinct roots of $\left(  1-p\right)
U_{x_{2}}\left(  p\right)  $. \ Thus equation~(\ref{ptp_sep}) holds.

Next we show that the pseudocode can be implemented in $\mathsf{PSPACE}$.
\ Observe first that the degrees of $U_{x},V_{x}$ are $2^{\operatorname*{poly}%
\left(  n\right)  }$, with rational coefficients having numerator and
denominator bounded by $2^{\operatorname*{poly}\left(  n\right)  }$.
\ Moreover, the coefficients of $U_{x},V_{x}$ are computable in
$\mathsf{PSPACE}$ from the coefficients of $Q_{x},R_{x}$, and these
coefficients are themselves $\mathsf{PSPACE}$-computable. \ To loop over the
elements of $\mathcal{P}$ as in the for-loops of the pseudocode, we can
perform an outer loop over $y\in\left\{  0,1\right\}  ^{n}$ and an inner loop
over the zeros of $U_{y}$. These zeros are indexed by Neff's algorithm
(Theorem~\ref{findroots}) and can be looped over with that indexing. \ The
algorithm of Theorem~\ref{findroots} may return duplicate roots, but these can
be identified and removed by comparing each root in turn to all previously
visited roots. \ For each pair of distinct zeros of $U_{y}$ differ in their
binary expansion to a sufficiently large $2^{\operatorname*{poly}\left(
n\right)  }$ number of bits (by Theorem~\ref{sepbound}), and
Theorem~\ref{findroots} allows us to compare such bits in polynomial space.

Similarly, if $\left(  y,p\right)  ,\left(  z,q\right)  \in\mathcal{P}$ then
we can determine in $\mathsf{PSPACE}$ whether $q\leq p$, as required. \ The
only remaining implementation step is to produce the value $r_{n}$ in
$\mathsf{PSPACE}$, in such a way that part (iii) of Lemma~\ref{findbias} is
satisfied. \ Given the value $p$ chosen by the inner loop, and the index
$i\leq2^{\operatorname*{poly}\left(  n\right)  }$, we need to produce the
$i^{th}$ bit of a value $r_{n}\in\left(  p,p+2^{-2^{\operatorname*{poly}%
\left(  n\right)  }}\right)  $, such that the binary expansion of $r_{n}$ is
identically zero for sufficiently large $j\geq h(n)=2^{\operatorname*{poly}%
\left(  n\right)  }$. But this is easily done, since we can compute any
desired $j^{th}$ bit of $p$, for $j\leq2^{\operatorname*{poly}\left(
n\right)  }$, in polynomial space.
\end{proof}%

\begin{figure}
[ptb]
\begin{center}
\includegraphics[
trim=0.000000in 3.010514in 0.201697in 0.804056in,
height=2.2018in,
width=5.2632in
]%
{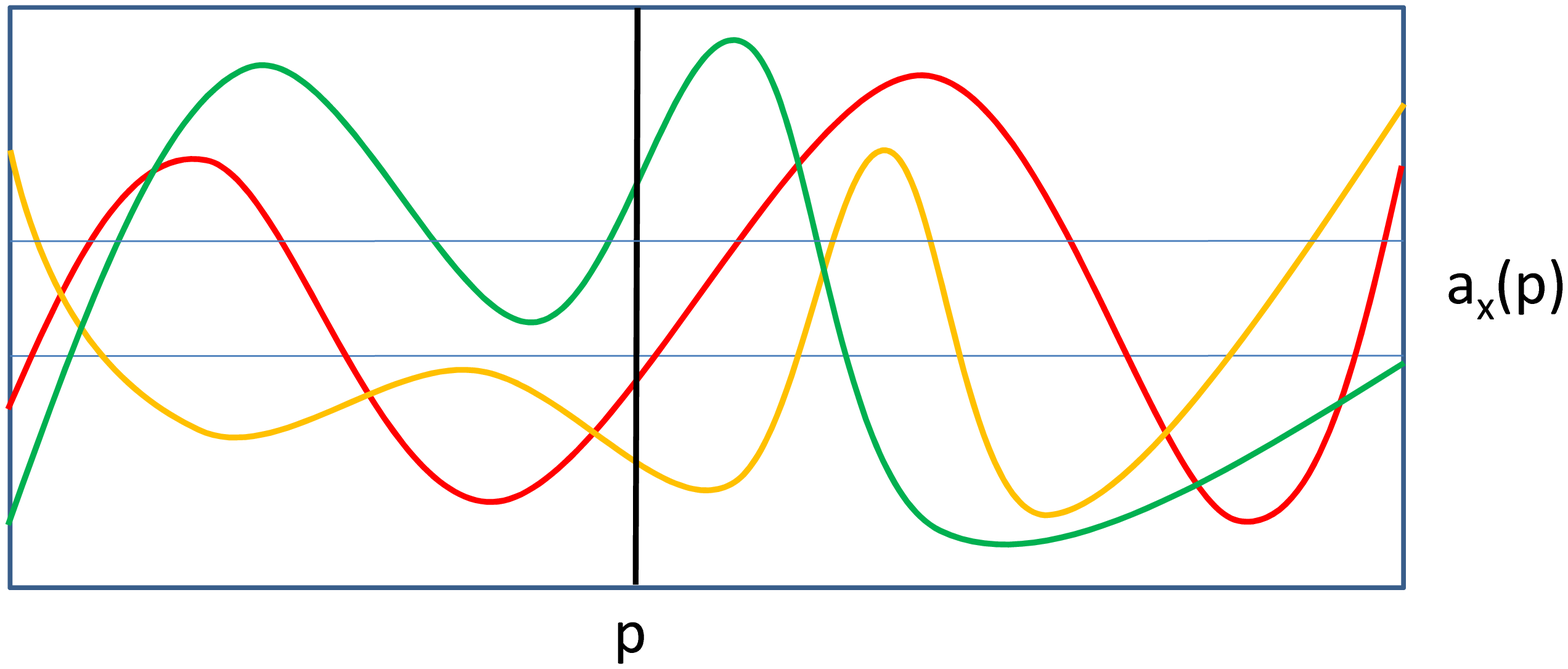}%
\caption{Graphical depiction of the proof of Theorem \ref{mainthm}, that
$\mathsf{BQPSPACE/coin}=\mathsf{PSPACE/poly}$. \ For each input $x\in\left\{
0,1\right\}  ^{n}$, the acceptance probability of the $\mathsf{BQPSPACE/coin}%
$\ machine is a rational function $a_{x}\left(  p\right)  $\ of the coin bias
$p$, with degree at most $2^{\operatorname*{poly}\left(  n\right)  }$. \ Such
a function can cross the $a_{x}\left(  p\right)  =2/5$\ or $a_{x}\left(
p\right)  =3/5$\ lines at most $2^{\operatorname*{poly}\left(  n\right)  }$
times. \ So even considering \textit{all} $2^{n}$\ inputs $x$, there can be at
most $2^{\operatorname*{poly}\left(  n\right)  }$\ crossings in total. \ It
follows that, if we want to specify whether $a_{x}\left(  p\right)  <2/5$ or
$a_{x}\left(  p\right)  >3/5$\ for all $2^{n}$\ inputs $x$ simultaneously, it
suffices to give only $\operatorname*{poly}\left(  n\right)  $ bits of
information about $p$ (for example, the total number of crossings to the left
of $p$).}%
\label{axp}%
\end{center}
\end{figure}

\section{Distinguishing Problems for Finite Automata\label{MODELS}}

The \textit{distinguishing problem}, as described in Section \ref{INTRO}, is a
natural problem with which to investigate the power of restricted models of
computation. The basic task is to distinguish a coin of bias $p$ from a coin
of bias $p+\varepsilon$, using a finite automaton with a bounded number of
states. \ Several variations of this problem have been explored~\cite{hc,
cover}, which modify either the model of computation or the mode of
acceptance. \ A basic question to explore in each case is whether the
distinguishing task can be solved by a finite automaton whose number of states
is independent of the value $\varepsilon$ (for fixed $p$, say).

Variations of interest include:

\begin{enumerate}
\item[(1)] \textit{Classical vs.\ quantum finite automata.} \ We showed in
Section~\ref{NULLIFY} that, in some cases, quantum finite automata can solve
the distinguishing problem where classical ones cannot.

\item[(2)] $\varepsilon$\textit{-dependent vs.\ }$\varepsilon$%
\textit{-independent automata.} \ Can a single automaton $M$ distinguish
$\$_{p}$\ from $\$_{p+\varepsilon}$\ for every $\varepsilon>0$, or is a
different automaton $M_{\varepsilon}$ required for different $\varepsilon$?

\item[(3)] \textit{Bias }$\mathit{0}$\textit{ vs.\ bias }$1/2$\textit{.} \ Is
the setting $p=0$ easier than the setting $p=1/2$?

\item[(4)] \textit{Time-dependent vs.\ time-independent automata.} \ An
alternative, \textquotedblleft nonuniform\textquotedblright\ model of finite
automata allows their state-transition function to depend on the
\textit{current time step} $t\geq0$, as well as on the current state and the
current bit being read. \ This dependence on $t$ can be arbitrary; the
transition function is not required to be computable given $t$.

\item[(5)] \textit{Acceptance by halting vs.\ }$\mathit{1}$\textit{-sided
acceptance vs.\ acceptance in the limit.} \ How does the finite automaton
register its final decision? \ A first possibility is that the automaton halts
and enters an $\left\vert \operatorname*{Accept}\right\rangle $\ state if it
thinks the bias is $p+\varepsilon$, or halts and enters a $\left\vert
\operatorname*{Reject}\right\rangle $\ state if it thinks the bias is $p$. \ A
second possibility, which corresponds to the model considered for most of this
paper, is that the automaton halts and enters an $\left\vert
\operatorname*{Accept}\right\rangle $\ state if it thinks the bias is
$p+\varepsilon$, but can reject by simply never halting. \ A third possibility
is that the automaton \textit{never} needs to halt. \ In this third model, we
designate some subset of the states as \textquotedblleft accepting
states,\textquotedblright\ and let $a_{t}$\ be the probability that the
automaton would be found in an accepting state, were it measured at the
$t^{th}$\ time step. \ Then the automaton is said to \textit{accept in the
limit} if $\liminf_{t\rightarrow\infty}(a_{1}+\ldots+a_{t})/t\geq2/3$, and to
\textit{reject in the limit} if $\limsup_{t\rightarrow\infty}(a_{1}%
+\cdots+a_{t})/t\leq1/3$. \ The automaton solves the distinguishing
problem\ if it accepts in the limit on a coin of bias $p+\varepsilon$, and
rejects in the limit on a coin of bias $p$.
\end{enumerate}

For almost every possible combination of the above, we can determine whether
the distinguishing problem can be solved by an automaton whose number of
states is independent of $\varepsilon$, by using the results and techniques
of~\cite{hc, cover} as well as the present paper. \ The situation is
summarized in the following two tables.%

\begin{align*}
&
\begin{tabular}
[c]{lllllll}%
\textbf{Classical case} & \multicolumn{6}{l}{\textit{Coin distinguishing
task}}\\
& $\frac{1}{2}$ vs.\ $\frac{1}{2}+\varepsilon$ &  &  & $0$ vs.\ $\varepsilon$
&  & \\
\textit{Type of automaton} & Halt & $1$-Sided & Limit & Halt & $1$-Sided &
Limit\\
Fixed & No & No & No & No & Yes (easy) & Yes\\
$\varepsilon$-dependent & No & No & No \cite{hc} & Yes (easy) & Yes & Yes\\
Time-dependent & No & Yes \cite{cover} & Yes & No & Yes & Yes\\
$\varepsilon$,time-dependent & Yes \cite{cover} & Yes & Yes & Yes & Yes & Yes
\end{tabular}
\\
&
\begin{tabular}
[c]{lllllll}%
\textbf{Quantum case} & \multicolumn{6}{l}{\textit{Coin distinguishing task}%
}\\
& $\frac{1}{2}$ vs.\ $\frac{1}{2}+\varepsilon$ &  &  & $0$ vs.\ $\varepsilon$
&  & \\
\textit{Type of automaton} & Halt & $1$-Sided & Limit & Halt & $1$-Sided &
Limit\\
Fixed & No & No (here) & ? & No & Yes & Yes\\
$\varepsilon$-dependent & \textbf{Yes (here)} & \textbf{Yes} & \textbf{Yes} &
Yes & Yes & Yes\\
Time-dependent & No & Yes & Yes & No (easy) & Yes & Yes\\
$\varepsilon$,time-dependent & Yes & Yes & Yes & Yes & Yes & Yes
\end{tabular}
\\
&
\end{align*}

Let us briefly discuss the possibility and impossibility results.

\begin{itemize}
\item[(1)] Hellman and Cover \cite{hc} showed that a classical finite
automaton needs $\Omega\left(  1/\varepsilon\right)  $\ states to distinguish
$p=1/2$\ from $p=1/2+\varepsilon$, even if the transition probabilities can
depend on $\varepsilon$\ and the automaton only needs to succeed in the limit.

\item[(2)] By contrast, Theorem \ref{qmnhc} shows that an $\varepsilon
$-dependent quantum finite automaton with only \textit{two} states can
distinguish $p=1/2$\ from $p=1/2+\varepsilon$ for any $\varepsilon>0$, even if
the automaton needs to halt.

\item[(3)] Cover \cite{cover} gave a construction of a $4$-state
\textit{time-dependent} (but $\varepsilon$-independent) classical finite
automaton that distinguishes $p=1/2$\ from $p=1/2+\varepsilon$, for any
$\varepsilon>0$, in the limit of infinitely many coin flips. \ This automaton
can even be made to halt in the case $p=1/2+\varepsilon$.

\item[(4)] It is easy to modify Cover's construction to get, for any
\textit{fixed} $\varepsilon>0$, a time-dependent, $2$-state finite automaton
that distinguishes $p=1/2$\ from $p=1/2+\varepsilon$ with high probability and
that halts. \ Indeed, we simply need to look for a run of $1/\varepsilon
$\ consecutive heads, repeating this $2^{1/\varepsilon}$\ times before
halting. \ If such a run is found, then we guess $p=1/2+\varepsilon$;
otherwise we guess $p=1/2$.

\item[(5)] If we merely want to distinguish $p=0$\ from $p=\varepsilon$, then
even simpler constructions suffice. \ With an $\varepsilon$-dependent finite
automaton, at every time step we flip the coin with probability $1-\varepsilon
$; otherwise we halt and guess $p=0$. \ If the coin ever lands heads, then we
halt and output $p=\varepsilon$. \ Indeed, even an $\varepsilon$%
\textit{-independent} finite automaton can distinguish $p=0$\ from
$p=\varepsilon$ in the $1$-sided model, by flipping the coin over and over,
and accepting if the coin ever lands heads.

\item[(6)] It is not hard to show that even a \textit{time-dependent, quantum}
finite automaton cannot solve the distinguishing problem, even for $p=0$
versus $p=\varepsilon$, provided that (i) the automaton has to halt when
outputting its answer, and (ii) the same automaton has to work for every
$\varepsilon$. \ The argument is simple: given a candidate automaton $M$, keep
decreasing $\varepsilon>0$\ until $M$ halts, with high probability, before
observing a single heads. \ This must be possible, since even if $p=0$\ (i.e.,
the coin \textit{never} lands heads), $M$ still needs to halt with high
probability. \ Thus, we can simply wait for $M$ to halt with high
probability---say, after $t$ coin flips---and then set $\varepsilon\ll1/t$.
\ Once we have done this, we have found a value of $\varepsilon$\ such that
$M$ cannot distinguish $p=0$\ from $p=\varepsilon$, since in both cases $M$
sees only tails with high probability.
\end{itemize}

\section{Open Problems\label{OPEN}}

\begin{enumerate}
\item[(1)] Our advice-coin computational model can be generalized
significantly, as follows. \ Let $\mathsf{BQPSPACE/dice}\left(  m,k\right)  $
be the class of languages decidable by a $\mathsf{BQPSPACE}$ machine that can
sample from $m$ distributions $\mathcal{D}_{1},\ldots,\mathcal{D}_{m}$, each
of which takes values in $\left\{  1,\ldots,k\right\}  $ (thus, these are
\textquotedblleft$k$-sided dice\textquotedblright). \ Note that
$\mathsf{BQPSPACE/coin}=\mathsf{BQPSPACE/dice}(1,2)$.

We conjecture that
\[
\mathsf{BQPSPACE/dice}\left(  1,\operatorname*{poly}\left(  n\right)  \right)
=\mathsf{BQPSPACE/dice}\left(  \operatorname*{poly}\left(  n\right)
,2\right)  =\mathsf{PSPACE/poly}.
\]
Furthermore, we are hopeful that the techniques of this paper can shed light
on this and similar questions.\footnote{Note that the distinguishing problem
for $k$-sided dice, for $k>2$, is addressed by the more general form of the
theorem of Hellman and Cover~\cite{hc}, while the distinguishing problem for
read-once branching programs was explored by Brody and Verbin~\cite{bv:coin}.}

\item[(2)] Not all combinations of model features in Section~\ref{MODELS} are
well-understood. \ In particular, can we distinguish a coin with
bias\ $p=1/2$\ from a coin with bias $p=1/2+\varepsilon$\ using a quantum
finite automaton, not dependent on $\varepsilon$,\ that only needs to succeed
in the limit?

\item[(3)] Given \textit{any} degree-$d$ rational function $a\left(  p\right)
$ such that $0\leq a\left(  p\right)  \leq1$\ for all $0\leq p\leq1$, does
there exist a $d$-state (or at least $\operatorname*{poly}\left(  d\right)
$-state) quantum finite automaton $M$\ such that $\Pr\left[  M\left(
\$_{p}\right)  \text{ accepts}\right]  =a\left(  p\right)  $?
\end{enumerate}

\section{Acknowledgments}

We thank Erik Demaine for suggesting the advice coins problem to us, and Piotr
Indyk for pointing us to the Hellman-Cover Theorem.

\bibliographystyle{plain}
\bibliography{thesis}

\end{document}